\documentclass[aps,reprint,floatfix]{revtex4-1}
\usepackage{amsmath,amssymb,graphicx,mathtools}
\usepackage{float}

\usepackage[T1]{fontenc}

\usepackage{colortbl}
\usepackage{multirow}

\usepackage[ruled]{algorithm2e}

\makeatletter
\DeclareRobustCommand{\cev}[1]{%
  \mathpalette\do@cev{#1}%
}

\makeatother

\newcounter{theoremcounter}
\stepcounter{theoremcounter}

\newtheorem{theorem}{Theorem}
\newtheorem{conjecture}{Conjecture}
\newtheorem{lemma}{Lemma}

\newenvironment{proof}[1][Proof]{\begin{trivlist}
\item[\hskip \labelsep {\bfseries #1}]}{\end{trivlist}}

\newcommand{\qed}{\hfill $\blacksquare$}

\newcommand{\bra}[1]{\langle #1|}
\newcommand{\ket}[1]{|#1\rangle}
\newcommand{\braket}[2]{\langle #1|#2\rangle}
\newcommand{\expval}[3]{\langle #1|#2|#3\rangle}

\makeatletter
\newcommand{\vast}{\bBigg@{4}}
\newcommand{\Vast}{\bBigg@{5}}
\makeatother

\begin{document}

\title{Improved Weak Simulation of Universal Quantum Circuits by Correlated \(L_1\) Sampling}
\author{Lucas Kocia}
\affiliation{Sandia National Laboratories, Livermore, California 94550, U.S.A.}
\begin{abstract}
Bounding the cost of classically simulating the outcomes of universal quantum circuits to additive error \(\delta\) is often called weak simulation and is a direct way to determine when they confer a quantum advantage. Weak simulation of the \(T\)+Clifford gateset is \(BQP\)-complete and is expected to scale exponentially with the number \(t\) of \(T\) gates. We constructively tighten the upper bound on the worst-case \(L_1\) norm sampling cost to next order in \(t\) from \(\mathcal O(\xi^t \delta^{-2})\) if \(\delta^2 \gg \xi^{-t}\) to \(\mathcal O((\xi^t{-}t) \delta^{-2} )\) if \(\delta^2 \gg (\xi^t -t)^{-1}\), where \(\xi^t = 2^{\sim 0.228 t}\) is the stabilizer extent of the \(t\)-tensored \(T\) gate magic state. We accomplish this by replacing independent \(L_1\) sampling in the popular SPARSIFY algorithm used in many weak simulators with correlated \(L_1\) sampling.
As an aside, this result demonstrates that the \(T\) gate magic state's approximate stabilizer state decomposition is not multiplicative with respect to \(t\), for finite values, despite the multiplicativity of its stabilizer extent.
This is the first weak simulation algorithm that has lowered this bound's dependence on finite \(t\) in the worst-case to our knowledge and establishes how to obtain further such reductions in \(t\).
\end{abstract}
\maketitle

Weak simulation is defined as the task of sampling the probabilities of universal quantum circuits to additive error. It is expected to require exponential resources on a classical computer since it is \(BQP\)-complete. Reducing the cost of classically simulating quantum computers~\cite{Harrow17} is necessary to characterize near-term noisy intermediate-scale quantum (NISQ) computers~\cite{Preskill18} that are rapidly growing in size and performance.

Universal quantum computation can be achieved using stabilizer states, the Clifford+\(T\) gateset, and Pauli measurement. An equivalent measurement-based formalism can be written in terms of stabilizer states, \(T\) gate magic states and Pauli measurements~\cite{Bravyi16_2}. Approximating outcomes samples to additive error \(\delta\) is equivalent to replacing the underlying probability distribution with one that is \(\delta\)-close to it and then sampling from this approximate distribution. This naturally splits up many weak simulation implementations into a ``sparsification'' step and a measurement step. The measurement step consists of taking idempotent projections, \(\expval{\psi}{\Pi}{\psi} = |\Pi\ket{\psi}|^2 \equiv |\psi'|^2\), and so is frequently called a ``normalization'' step instead.

The SPARSIFY algorithm introduced~\cite{Bravyi16_1} a method of generating an \(L_1\) sparsification of a given state \(\psi\) to \(\delta\) additive error with \(\mathcal O(2^{\sim 0.228 t} \delta^{-2})\) stabilizer states, which is asymptotically optimal as \(t \rightarrow \infty\) and \(\delta \rightarrow 0\) (see Lemma \(2\) in~\cite{Bravyi16_1}). As a result, the authors conjectured the following lower bound:
\begin{conjecture}
  \label{conj:lowerbound}
  Any approximate stabilizer decomposition of \(T^{\otimes t}\) that achieves a constant approximation error must use at least \(\Omega(2^{\sim 0.228 t})\) stabilizer states.
\end{conjecture}
This approximated state's inner product must then be sampled under random Pauli measurements to complete a weak simulation algorithm. The full weak simulation cost is the number of stabilizer states produced by SPARSIFY multiplied by \(\mathcal O(t^3 \delta^{-2})\).

Subsequent works~\cite{Howard18,Seddon20,Pashayan21} almost all use the SPARSIFY algorithm or a similar sparsification method.
Improvements have included an extension of the method to diagonal states (other than the \(T\) magic state)~\cite{Howard18} and mixed states~\cite{Seddon20}, constant factor improvements~\cite{Howard18}, a decrease in the power of \(\delta\) cost for magic states~\cite{Seddon20}, better performance when the values of the sampled probabilities are in certain regimes~\cite{Pashayan21}, and an extension to Born probabilities~\cite{Pashayan21}. 

Nevertheless, these methods have all saturated the asymptotic conjectured lower bound w.r.t. \(t\) even when they are not in the asymptotic limit of \(t \rightarrow \infty\) and \(\delta \rightarrow 0\); they all require \(\mathcal O (2^{\sim 0.228 t})\) stabilizer states in the worst case.

Since the lower bound given by Conjecture~\ref{conj:lowerbound} is an asymptotic bound, there is no reason to consider it limiting for finite \(t\) and \(\delta\). Indeed, the finite regime is the most useful for practical simulations and validations of near-term devices. Non-asymptotic reductions in \(t\) can greatly increase the size of universal quantum circuits that are simulatable by today's classical computers and thereby change when they confer quantum advantage.

Here we introduce the first such reduction in \(t\) and demonstrate its practical usefulness for finite-sized circuits. Since the reduction occurs in the SPARSIFY algorithm used by many contemporary weak simulators, it can be implemented in current applications with minimal change and improve their performance. The key idea is replacing independent \(L_1\) sampling with correlated \(L_1\) sampling.

To begin, we define the general family of diagonal states that we want to approximate. Following~\cite{Bravyi16_1,Bravyi16_2}, we define the \(t\)-tensored state:
\begin{equation}
\label{eq:diagstate}
  \ket{D_\phi^{\otimes t}} \equiv (2 \nu)^{-t} \sum_{x \in \mathbb F_2^t} \ket{\tilde x_1 \otimes \cdots \otimes \tilde x_t},
\end{equation}
where 
\begin{eqnarray}
\ket{\tilde 0} \equiv \frac{i}{\sqrt{2}} (-i + e^{-\pi i/4}) (-i + e^{i \phi}) \ket{0},\\
\ket{\tilde 1} \equiv \frac{i}{\sqrt{2}} (1 + e^{-\pi i/4}) (1 - e^{i \phi}) \frac{1}{\sqrt{2}} (\ket{0} + \ket{1}),
\end{eqnarray}
and \(\nu \equiv \cos \pi/8\). \(\ket{D_\phi^{\otimes t}}\) has \(L_1\) norm squared, when considered over the set of stabilizer states, that is minimized by this stabilizer decomposition into \(\ket{\tilde 0}\) and \(\ket{\tilde 1}\): 
\begin{equation}
\left(\sum_{x \in \mathbb F_2^t} |c_{\tilde x}|\right)^2 = \left(\sqrt{1 - \sin\phi} + \sqrt{1 - \cos\phi}\right)^{2t} \equiv \xi_\phi^t.
\end{equation}
Calling \(\xi_\phi^t\) an \(L_1\) norm is an abuse of the term, since stabilizer states are overcomplete. As a result \(\xi_\phi^t\) is often called the \emph{stabilizer extent}, and is more accurately defined as the minimum value of \(\|c\|_1\) over all stabilizer decompositions.

\(\ket{D_{\pi/4}^{\otimes t}} = \ket{H^{\otimes t}}\) 
for \(\ket H \equiv e^{-i \pi/8} S H \ket T\) and
\begin{equation}
\ket{T} = \frac{1}{\sqrt{2}} (\ket 0 + \sqrt{i} \ket 1),
\end{equation}
 where \(S\) and \(H\) are the Clifford gates phase shift and Hadamard, respectively. \(\xi_{\pi/4}^t = \nu^{-2t} = 2^{\sim 0.228 t}\), which sets the exponential scaling for sampling this magic state. Since the \(H\) magic state is related to the \(T\) magic state by a Clifford unitary under which stabilizer states are closed, its stabilizer state approximation cost is the same.

Before we get into generating an approximation to these states, we first need to establish a useful result:

 \begin{lemma}[\(t\) Stabilizer States with \(t/2\) Bit Flips]
   \label{le:bitflip}
   Given a \(t\)-bitstring, there exist \(t\) additional different \(t\)-bitstrings such that every pair of bitstrings differs by at least \(t/2\) bits.
 \end{lemma}
 \begin{proof}
   We examine powers of \(2\). WLOG we can always assume that the given bitstring consists of all \(1\)s since any additional bitstrings to this can be generalized to any other bitstring by XORing it.

   Let \(t=2^k\). Consider iterative splitting of the \(2^k\) bitstring by binary tree with layers \(0 \le i \le k-1\) (see Table~\(1\) for examples). We define the base layer \(k=0\) to consist of two bitstrings, \(\alpha \cdots \alpha\) and \(\beta \cdots \beta\), where \(\alpha \equiv 01\) and \(\beta \equiv 10\). Given a \(t\)-bitstring of all \(1\)s, it is clear that these two bitstrings differ from it by \(t/2\) bitflips and from each other by \(t>t/2\) bitflips.

In every subsequent layer we consider \(t\)-bitstrings that are evenly split into \(2 \le 2^i \le 2^{k-1}\) contiguous blocks of bits that we treat all together when assigning values. We assign values \(\alpha\) and \(\beta\) to the blocks \emph{such that} the number of \(\alpha\)s and \(\beta\)s are even in every pair of blocks corresponding to a larger block a layer above and at least half the assignments differ from every other bitstring in the same layer.

This means that every layer \(1 \le i \le k-1\) will have \(t/2\) \(\alpha\)s and \(t/2\) \(\beta\)s. Hence, they will differ from the given bitstring of all \(1\)s by \(t/2\) bit flips. This also means that they will differ from bitstrings in the zeroth layer consisting of only \(\alpha\)s or \(\beta\)s by \(t/2\) bit flips.

Moreover, since the nature of this binary tree splitting converts blocks from higher levels into evenly split subblocks with the same number of \(\alpha\)s and \(\beta\)s where before there were only \(\alpha\)s or \(\beta\)s, every bitstring differs from those of other layers by \(t/2\) bitflips.

Within a layer, by construction, the bitstrings differ from each other by at least \(t/2\) bitflips. Trivially, the first layer (\(i=1\)) consists of two bitstrings. Subsequent layers consist of twice as many bitstrings as their preceding layers since they can be considered as the result of the same bitflips performed on the given bitstring as the layer above, but over twice as many bits. It follows that the \(i\)th level consists of \(2^i\) bitstrings. This means that there is a total of \(2 + \sum_{i=1}^{k-1} 2^i = 2^k = t\) bitstrings.

We have therefore \(t\) additional total of bitstrings that differ from the given bitstrings of all \(1\)s and each other by at least \(t/2\) bitflips.~\qed
 \end{proof}
\begin{table}
\begin{tabular}{c|c|c|c}
\(t\) & given bitstring & depth \(i\) & \(t\) additional bitstrings\\
\hline
\(2\) & \(11\) & \(0\) & \(\alpha\), \(\beta\) \\
\hline
\multirow{2}{*}{\(4\)} & \multirow{2}{*}{\(1111\)} & \(0\) & \(\alpha \alpha\), \(\beta \beta\)\\
\cline{3-4} & & \(1\) & {\(\alpha \beta\), \(\beta \alpha\)}\\
\hline
\multirow{3}{*}{\(8\)} & \multirow{3}{*}{\(11111111\)} & \(0\) & \(\alpha \alpha \alpha \alpha\), \(\beta \beta \beta \beta\)\\
  \cline{3-4}
  & & \(1\) & \(\alpha \alpha \beta \beta\), \(\beta \beta \alpha \alpha\) \\
  \cline{3-4}
  & & \(2\) & \(\alpha \beta \beta \alpha\), \(\beta \alpha \alpha \beta\), \(\alpha \beta \alpha \beta\), \(\beta \alpha \beta \alpha\)\\
\hline
\multirow{8}{*}{\(16\)} & \multirow{8}{*}{\(1111111111111111\)} & \(0\) & \(\alpha \alpha \alpha \alpha \alpha \alpha \alpha \alpha\), \(\beta \beta \beta \beta \beta \beta \beta \beta\)\\
  \cline{3-4}
      & & \(1\) & \(\alpha \alpha \alpha \alpha \beta \beta \beta \beta\), \(\beta \beta \beta \beta \alpha \alpha \alpha \alpha\)\\
  \cline{3-4}
      & & \multirow{2}{*}{\(2\)} & \(\alpha \alpha \beta \beta \alpha \alpha \beta \beta\), \(\beta \beta \alpha \alpha \beta \beta \alpha \alpha\),\\
      & & & \(\alpha \alpha \beta \beta \beta \beta \alpha \alpha\), \(\beta \beta \alpha \alpha \alpha \alpha \beta \beta\)\\
  \cline{3-4}
      & & \multirow{4}{*}{\(3\)} & \(\alpha \beta \alpha \beta \alpha \beta \alpha \beta\), \(\alpha \beta \alpha \beta \beta \alpha \beta \alpha\),\\
      & & & \(\beta \alpha \beta \alpha \alpha \beta \alpha \beta\), \(\beta \alpha \beta \alpha \beta \alpha \beta \alpha\),\\
      & & & \(\alpha \beta \beta \alpha \alpha \beta \beta \alpha\), \(\beta \alpha \alpha \beta \beta \alpha \alpha \beta\),\\
      & & & \(\alpha \beta \beta \alpha \beta \alpha \alpha \beta\), \(\beta \alpha \alpha \beta \alpha \beta \beta \alpha\)\\
\end{tabular}
\label{tab:binarytree}
\caption{\(\alpha \equiv 01\) and \(\beta \equiv 10\). The additional bitstrings can be used as XOR masks to generate the appropriate additional bitstrings for given bitstrings other than all \(1\)s.}
\end{table}

An algorithm that generates \(t\) additional bitstrings that differ by at least \(t/2\) bitflips, such as those given in Table~\(1\), is given in Algorithm~\ref{alg:bitflips}.

\begin{algorithm}
  \KwData{$k$ such that $t=2^{k}$.}
  \KwResult{\emph{bitstring} array.}
  \Begin{
    \emph{bitstrings} \(\leftarrow\, \{\alpha\cdots\alpha, \, \beta\cdots\beta\}\)\;
    \emph{masks} \(\leftarrow\, \{2^{2^k}\}\)\;
    \For{treedepth \(\leftarrow 1\) \KwTo \(k-1\)}{
      \emph{levelmask} \(\leftarrow\) \(2^{2^{k-treedepth}}-1\)\;
      \For{levelmaskdepth \(\leftarrow 2\) \KwTo \(2^{treedepth-1}\)}{
        \emph{levelmask} \(\leftarrow\) \emph{levelmask} \(+\) \(2^{2^{k-treedepth+1}} \times\)\emph{levelmask}\;
      }
      \(\text{\emph{bitstring}} \leftarrow (\alpha\cdots\alpha\, \mbox{XOR}\, \text{\emph{levelmask}})\)\;
      \For{mask \(\leftarrow\) Cartesian products of elements in masks}{
        add \((\text{\emph{bitstring}} \, \mbox{XOR}\, \text{\emph{mask}})\) to \emph{bitstrings}\;
      }
      add \emph{levelmask} to \emph{masks}\;
    }
  }
  \caption{Generate additional bitstrings that differ from the \(t\)-bitstring of all \(1\)s by at least \(t/2\) bitflips.}
  \label{alg:bitflips}
\end{algorithm}

We introduce a constructive upper bound that is lower than Conjecture~\ref{conj:lowerbound}'s in the finite \(t\) case:
\begin{theorem}[Lower Bound in \(t\) for SPARSIFY]
  \label{th:sparsifyupperbound}
  The SPARSIFY procedure introduced by Bravyi~\emph{et al}.~\cite{Howard18} creates a \(\delta\)-approximate stabilizer decomposition of \(H^{\otimes t}\) with \(\mathcal O((2^{\sim 0.228 t}{-}t) \delta^{-2} )\) stabilizer states for \(t\) sufficiently large such that \(\delta^2 \gg (\xi^t_{\pi/4} -t)^{-1}\).
\end{theorem}
\begin{proof}
Following~\cite{Howard18}, we define additive error
\begin{equation}
\|\ket{D_\phi^{\otimes t}} - \ket{\psi}\| \le \delta,
\end{equation}
where \(\|\psi\| \equiv \sqrt{\braket{\psi}{\psi}}\).

\(\ket{\psi}\) is the sparsified \(k\)-term approximation to \(\ket{D_\phi(t)}\) given by
\begin{equation}
  \ket \psi = \frac{\|c\|_1}{k} \sum_{i=1}^k \ket{\omega_i},
\end{equation}
where each \(\ket{\omega_i}\) is independently chosen randomly so that it is a normalized stabilizer state \(\ket{\omega_i} = c_i/|c_i| \ket{\varphi_i}\) with probability \(p_i = |c_i|/\|c\|_1\). We define a random variable \(\ket \omega\) that is equal to \(\ket{\omega_i}\) with probability \(p_i\). Then
\begin{equation}
  \mathbb E (\ket{\omega}) = \ket {\psi}/\|c\|_1.
\end{equation}

By construction,
\begin{equation}
\mathbb E(\braket{\psi}{D_\phi^{\otimes t}}) = \mathbb E(\braket{D_\phi^{\otimes t}}{\psi}) = 1.
\end{equation}

The number of stabilizer states in the approximation is \(k\) and
\begin{eqnarray}
  \mathbb E (\| \ket{D_\phi^{\otimes t}} - \ket{\psi}\|^2) &=& \mathbb E( |\braket{D_\phi^{\otimes t}}{D^{\otimes t}}| ) - \mathbb E( |\braket{D_\phi^{\otimes t}}{\psi}| ) -\nonumber\\
  \label{eq:kdependence}
  && \mathbb E( |\braket{\psi}{D_\phi^{\otimes t}}| ) + \mathbb E( |\braket{\psi}{\psi}| ) \nonumber\\
  &\le& \frac{\xi_\phi^t}{k} - \frac{\gamma}{k},
\end{eqnarray}
where we simplified
\begin{eqnarray}
  \mathbb E(\braket{\psi}{\psi} ) &=& \sum_i^k \frac{\|c\|_1^2}{k^2} \mathbb E(\braket{\omega_i}{\omega_i}) + \sum_{i\ne j}^k \frac{\|c\|_1^2}{k^2} \mathbb E(\braket{\omega_i}{\omega_j}) \nonumber\\
  &=& \frac{\|c\|_1^2}{k} \mathbb E(|\braket{\omega}{\omega}|) + \sum_{i\ne j}^k \frac{\|c\|_1^2}{k^2} \mathbb E(\braket{\omega_i}{\omega_j}) \nonumber\\
  &\le& \frac{\|c\|_1^2}{k} + 1 - \frac{\gamma}{k}.
\end{eqnarray}

Eq.~\ref{eq:kdependence} is less than or equal to \(\delta^2\) when \(k = (\xi_\phi^t - \gamma) \delta^{-2}\).

If \(\ket{\omega_i}\) are independent and identically distributed (i.i.d.) stabilizer states then \(\sum_{i\ne j}^k \|c\|_1^2 \mathbb E(\|\braket{\omega_i}{\omega_j}) = \sum_{i\ne j}^k | \mathbb E(\bra{\psi}) \mathbb E(\ket{\psi}) | = k (k-1)\) and so \(\gamma = 1\). As a result, since \(1 \ll \xi_\phi^t\) as \(t\) increases, it was neglected in previous characterizations~\cite{Bravyi16_1}. 

However, \(\gamma\) can become significant if \(\ket{\omega_i}\) are not i.i.d. and \(\sum_{i\ne j}^k \|c\|_1^2 \mathbb E(\|\braket{\omega_i}{\omega_j}) \ne \sum_{i\ne j}^k | \mathbb E(\bra{\psi}) \mathbb E(\ket{\psi}) |\).

In particular, let us consider sampling the \(H\) magic state \(\ket{D_{\pi/4}^{\otimes t}} = \ket{H^{\otimes t}}\). In this case, its minimal \(L_1\) stabilizer state decomposition consists of a uniform superposition over \(\ket{\tilde 0}\) and \(\ket{\tilde 1}\), where \(|\braket{\tilde 0}{\tilde 1}| = 2^{-\frac{1}{2}}\). In the SPARSIFY algorithm, \(t\)-bit strings consisting of these stabilizer states are uniformly sampled to approximate \(\ket{H^{\otimes t}}\). 

By Lemma~\ref{le:bitflip}, let us supplement this set of \(t\)-bit strings with the \(t\) \(t\)-bit strings that differ from every uniformly sampled state and each other by at least \(t/2\) bit flips (referring to the tilde basis). It follows that these \((t{+}1)\) stabilizer states have inner products of \(\le 2^{-\frac{t/2}{2}}\).

\begin{figure}[h]
\input{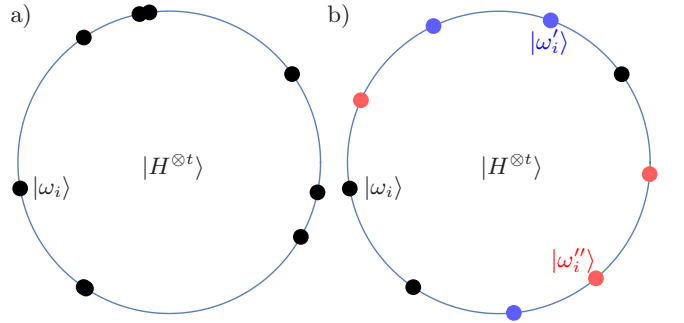}
\label{fig:cartoon_sketch}
\caption{Sketch of the different ensembles produced by (a) independent and (b) correlated \(L_1\) sampling. In this example, (a) there are nine independently sampled states, \(\{\ket{\omega_i}\}_i\), from the uniform distribution on the unit circle (the ``stabilizer state space''), which when considered as real vectors on \(\mathbb R^2\) have an expectation value close to \(\ket{H^{\otimes t}}\) at the origin. This ensemble can be transformed to a (b) correlated one, by supplementing the first three states with two \(\frac{t}{2}\)-bitflipped versions, \(\ket{\omega_i'}\) and \(\ket{\omega_i''}\) (blue and red), which are therefore far away and equidistant to each other on the unit circle, and discarding the rest. The expectation value of this ensemble is closer to \(\ket{H^{\otimes t}}\) but converges more slowly.}
\end{figure}

Since the ensemble consisting of uniformly sampled \(t\)-bit strings \(\ket{\omega_i}\) satisfies
\begin{equation}
\mathbb E(\braket{\psi}{H^{\otimes t}}) = \mathbb E(\braket{H^{\otimes t}}{\psi}) = 1,
\end{equation}
it follows that the \(t\) other ensembles consisting of the \(i\)th state with at least \(t/2\) bits flipped (\(i \in \{1, \ldots, t\}\)) compared to the uniformly sampled states, also satisfy this property. Therefore, the full ensemble produced by adding together these \((t{+}1)\) ensembles satisfies this property too.

However, taken together, these are no longer i.i.d. stabilizer states. In particular, for a given \(\bra{\omega_i}\), there exist at least \(t\) \(\ket{\omega_{f_i(j)}}\) such that \(|\braket{\omega_i}{\omega_{f_i(j)}}| \le 2^{-\frac{t/2}{2}}\). Hence,
\begin{eqnarray}
\label{eq:offdiagnorm}
\sum_{i\ne j}^k \|c\|_1^2 \mathbb E(\braket{\omega_i}{\omega_j}) &=& \sum_{i}^k \sum_{\substack{j\\i\ne j}}^{k-t} \|c\|_1^2 \mathbb E(\bra{\omega_i})\mathbb E(\ket{\omega_j}) \\
                                                                 && + \sum_{i}^k \sum_{j}^t \|c\|_1^2 \mathbb E(\bra{\omega_i})\mathbb E(\ket{\omega_{f_i(j)}}) \nonumber\\
                                                                 &\le& k (k-1-t) + \|c\|_1^2  2^{-\frac{t/2}{2}} t \\
                                                                 &=& k(k-\gamma),
\end{eqnarray}
where \(\|c\|^2_1 = \xi_{\pi/4}^t = 2^{\sim 0.228 t}\) and so \(\gamma = 1+(1-1/2^{\sim 0.02t})t\).

Therefore, given that at least \(k = (\xi_\phi^t - \gamma) \delta^{-2}\) stabilizer states are necessary to sample this state to \(\delta\) additive error, the SPARSIFY procedure creates a \(\delta\)-approximate stabilizer decomposition of \(H^{\otimes t}\) with \(\mathcal O((2^{\sim 0.228 t}{-}1{-}(1-1/2^{\sim 0.022 t}) t) \delta^{-2} )\) stabilizer states.

This more efficiently approximated state comes at the expense of its convergence probability, or sparsification tail bound. Following the same reasoning as in the proof of Lemma \(7\) of~\cite{Howard18},
\begin{eqnarray}
  &&\Pr \left[ \|H^{\otimes t} - \psi\|^2 \le \braket{\psi}{\psi} - 1 + \delta^2 \right] \nonumber\\
  \label{eq:tailbound}
  &&\ge 1 {-} 2 \exp \left( -\frac{\delta^2 \xi^t_{\pi/4}}{8 } {+} \frac{\gamma \delta^2}{8} \right) \\
  \label{eq:tailbound_final}
  &&= 1 {-} 2 \exp \left( -\frac{\delta^2 \xi^t_{\pi/4} }{8 } {+} \frac{(1 {+} (1{-}1/2^{\sim 0.02t}) t) \delta^2}{8} \right).
\end{eqnarray}

Therefore, given that \(\delta^2 \gg (\xi^t_{\pi/4} -t)^{-1}\), if post-selection is performed to discard samples that produce \(\braket{\psi}{\psi} - 1 \gg \delta^2\) (a rare event if this first condition is met) or \(\braket{\psi}{\psi}\) is approximated to relative error using the FASTNORM algorithm~\cite{Howard18} (which scales linearly with \(k\)), then the states \(\psi\) are generated with \(\mathbb E (\| \ket{H^{\otimes t}} - \ket{\psi}\|^2) < \delta^2 \) and consist of \(\mathcal O((2^{\sim 0.228 t} {-} t) \delta^{-2})\) stabilizer states.
\qed
\end{proof}

A sketch of the key idea used in the proof of Theorem~\ref{th:sparsifyupperbound} is shown in Figure~\ref{fig:cartoon_sketch}1. Independently sampled states with expectation value \(\ket{H^{\otimes t}}\) are replaced with a smaller subset that are supplemented with bit-flipped states. The resultant correlated distribution has an expectation value closer to \(\ket{H^{\otimes t}}\), but it converges to it more slowly.

Some polynomial factors in \(t\) are not included in the scaling cost, \(\mathcal O((2^{\sim 0.228 t} {-} t) \delta^{-2})\), of SPARSIFY. Moreover, there is a possible additional polynomial cost in correlated sampling from generating the bit-flipped supplemental states (such as using Algorithm~\ref{alg:bitflips}) compared to independent sampling. We claim that these changes in polynomial factors are negligible. This claim is supported by the scaling observed in the practical runtime of SPARSIFY plotted in Figure~\ref{fig:normscaling}2. A decrease in runtime is observed for correlated sampling that is lower bounded by proportionality to the fewer number of stabilizer states \(k = (\xi_\phi^t - \gamma) \delta^{-2}\) it generates.

\begin{figure}[h]
\input{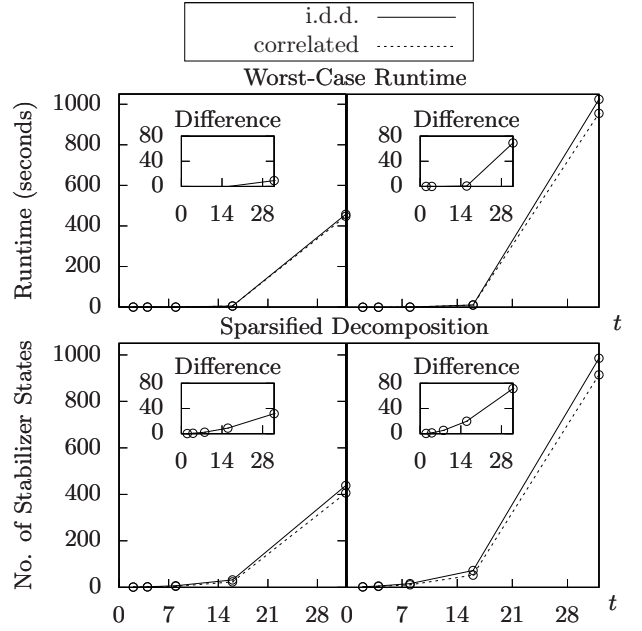}
\label{fig:normscaling}
\caption{Plots are for additive error (left) \(\delta = 0.6\) and (right) \(\delta = 0.4\) over \(100\) runs. The worst-case runtime of calculating the norm of \(\psi\) where \(|\psi {-} T^{\otimes t}| \le \delta^2\) is plotted at the top and the corresponding number of stabilizer states sampled in the sparsified decomposition with i.i.d. (solid curve) and correlated sampling (dashed curve) is plotted at the bottom. \(\psi\) is generated using the SPARSIFY algorithm and the norm is calculated using the FASTNORM algorithm of~\cite{Bravyi16_1} (using \(1000\) random stabilizer states to calculate the relative error). [Insets: The difference between the i.i.d. sampling and the correlated sampling curves.]}
\end{figure}

The statistical distribution of sparsified decompositions from independent sampling and correlated sampling are compared over \(1000\) numerical runs in Figure~\(3\). The expected value of the state generated by correlated sampling is closer to the desired state (middle of Figure~\(3\)). As a result, the standard deviation of the norm of the states generated by correlated sampling is larger (bottom of Figure~\(3\)) denoting poorer convergence, as expected. Hence, it is advantageous to use correlated sampling when \(\delta^2 \gg (\xi^t_{\pi/4} -t)^{-1}\) to obtain the same convergence probability as independent sampling does at \(\delta^2 \gg \xi^{-t}_{\pi/4}\). 

At small \(t\) a small-number effect occurs since the number of stabilizer states used in correlated sampling is set to the nearest multiple of \(t\) greater than or equal to \((\xi_\phi^t-\gamma)\delta^{-2}\) in practice. As a result, at small \(t\) more states are sampled than required and this produces a lower expectation value and standard deviation than expected. The factor of \(-1/2^{\sim 0.02 t}\) in Eq.~\ref{eq:tailbound_final} also reduces the standard deviation at low \(t\).

\begin{figure}[h]
\input{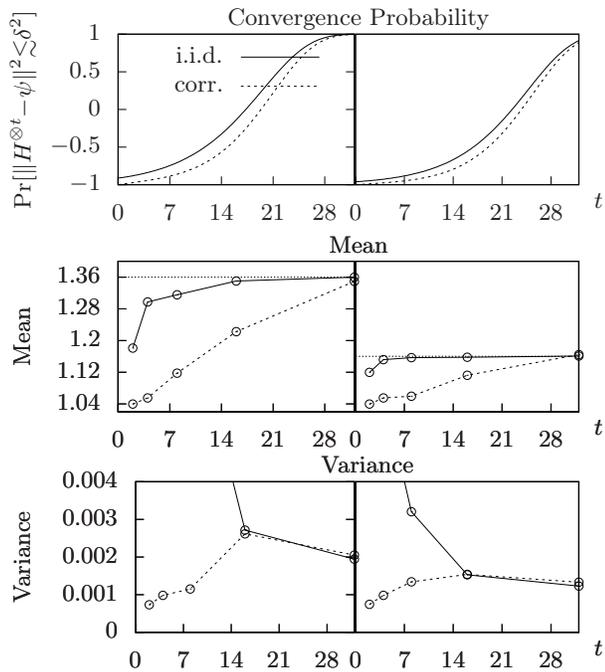}
\label{fig:meansandstdevs}
\caption{Plots are for additive error (left) \(\delta = 0.6\) and (right) \(\delta = 0.4\) over \(1000\) runs. At the top is plotted the convergence probability lower bound (or sparsification tail bound) given by Eq.~\ref{eq:tailbound_final} for i.i.d. sampling (solid curve) and correlated sampling (dashed curve). In the middle is plotted the mean of \(|\braket{\psi}{\psi}|^2\) for both i.i.d. and correlated sampling to \(\delta\) error (\(1+\delta^2\) is denoted by the dotted horizontal line). At the bottom is plotted the standard deviation of \(|\braket{\psi}{\psi}|^2\) of the sparsified samples, which can be interpreted as a measure of the convergence probability. The standard deviation converges more slowly for large \(t\) under correlated sampling than under independent sampling. This agrees with the requirement that \(\delta^2 \gg \xi^{-t}\) and \(\delta^2 \gg(\xi^t - t)^{-1}\) for i.i.d. and correlated sampling to exhibit \(\mathcal O(\xi^t \delta^{-2})\) and \(\mathcal O((\xi^t{-}t) \delta^{-2} )\) scaling, respectively, with the same probability. A small-\(t\) effect can be seen where the mean and standard deviation of the correlated samples is lower than that of the i.i.d. samples at \(t \lesssim 16\) as explained in the main text. The number of i.i.d. samples and correlated samples generated at particular \(t\)-values is shown in Fig.~\(2\).}
\end{figure}

This method can be extended to produce higher powers of \(t\) in \(\gamma\) and thereby improve performance further. In the proof of Theorem~\ref{th:sparsifyupperbound}, the source of the linear power in \(\gamma\) is due to correlated \(t\)-wise \(L_1\) sampling; every i.i.d. sampled state is supplemented with \(t\) samples with a known relative absolute inner product given by Lemma~\ref{le:bitflip}. However, it is easy to show that the number of mutually \({\ge} t/2\)-bitflipped states is larger than \(t\) and the number of states given by Lemma~\ref{le:bitflip} is a loose lower bound. The number of supplemented states can be increased to \(t^m\), for \(m>1\), limited by the minimal number \(k = (\xi_\phi^t - \gamma) \delta^{-2}\) of stabilizer states needed and the existence of states with the minimum number of mutual bitflips desired. This will add a corresponding power of \(t^m\) instead of \(t\) in \(\gamma\). It is also possible to extend \(\gamma\) to higher powers \(t^m\) for fixed \(k\) by supplementing every i.i.d. sampled state with \(t\) samples that have different relative absolute inner products. In both of these cases, doing so would increase the sparsification tail bound of \(\Pr \left[ \|H^{\otimes t} - \psi\|^2 \le \braket{\Omega}{\Omega} - 1 + \delta^2 \right]\) further. This would decrease the rate of convergence, and so would require \(\delta^2 \gg (\xi^t_{\pi/4} - t^{m})^{-1}\) for the improvement in scaling to outperform independent sampling to the same convergence probability. However, it is not clear how many such appropriately bit-flipped supplemental states exist given a bitstring and, therefore, it is not clear how large \(m\) of a reduction \(t^m\) in \(\gamma\) it is possible to accomplish. We leave this unresolved for future study.

A similar approach will also extend this method of correlated \(L_1\) norm sampling to any of the other diagonal states expressed by Eq.~\ref{eq:diagstate}. Such a treatment would differ only in that the distribution of \(t\) bit-flipped bit strings would be sampled from the non-uniform distribution given by Eq.~\ref{eq:diagstate} for \(\phi \ne \frac{\pi}{4}\).

Though the stabilizer extent \(\xi_\phi\) of one-, two-, and three-qubit states is multiplicative, general states do not have multiplicative stabilizer extent~\cite{Heimendahl21}. This introduces the peculiar notion that \(L_1\) sampling, which is upper bounded by the stabilizer extent (see Lemma \(2\) in~\cite{Bravyi16_1}), cannot do better than \(\mathcal O(2^{\sim 0.228 t})\) for the \(T\) gate magic state, but that you can always find a more optimal stabilizer decomposition for higher values of \(t\) for other states such that their worst-case scaling improves.

The results shown here may resolve this peculiarity. Namely, they show that \(L_1\) sampling is only asymptotically bounded by the stabilizer extent and that, for finite \(t\) values, an improvement can be found. This means that the scaling of the \(L_1\) sampling cost of one-, two-, and three- qubit magic states may behave similarly to the scaling of general states.

In conclusion, we show how to lower the finite \(t\) scaling cost of the popular SPARSIFY algorithm used in weak simulation of the \(T\)+Clifford gateset from \(\mathcal O(2^{\sim 0.228 t} \delta^{-2} )\) to \(\mathcal O((2^{\sim 0.228 t}{-} t) \delta^{-2} )\). We accomplish this by replacing its i.i.d. \(L_1\) sampling with correlated \(L_1\) sampling and we numerically demonstrate that this scaling reduction holds after including hidden prefactors polynomial in \(t\). We explain how further reductions in powers of \(t\) can be obtained with this method. To our knowledge, this is the first weak simulation algorithm that has lowered this bound's dependence on finite \(t\) in the worst-case.

\noindent---\newline
This material is based upon work supported by the U.S. Department of Energy, Office of Science, Office of Advanced Scientific Computing Research, under the Accelerated Research in Quantum Computing program.
Sandia National Laboratories is a multimission laboratory managed and operated by National Technology \& Engineering Solutions of Sandia, LLC, a wholly owned subsidiary of Honeywell International Inc., for the U.S. Department of Energy's National Nuclear Security Administration under contract {DE-NA0003525}. This paper describes objective technical results and analysis. Any subjective views or opinions that might be expressed in the paper do not necessarily represent the views of the U.S. Department of Energy or the United States Government.

\acknowledgments
The author thanks Mohan Sarovar for helpful discussions in the process of this research.

\bibliography{biblio}{}
\bibliographystyle{unsrt}

\appendix

\end{document}